\documentclass[10pt]{article}

\usepackage{paperlighter}

\usepackage{microtype}
\usepackage{graphicx}
\usepackage{subfigure}
\usepackage{booktabs} 
\usepackage{rotating}
\usepackage[title]{appendix}
\usepackage{gnuplottex}
\usepackage{nicefrac}
\usepackage{url}
\usepackage{xcolor}
\usepackage{amsmath}
\usepackage{amssymb}
\usepackage{afterpage,lipsum}
\usepackage[numbers]{natbib}

\usepackage{amsmath}
\usepackage{amssymb}
\usepackage{mathtools}
\usepackage{amsthm}
\usepackage{soul}
\usepackage{geometry}

\usepackage{mathtools}

\newcommand{\probP}{\mathbb P}
\newcommand{\exP}{\mathbb E}

\usepackage[capitalize,noabbrev]{cleveref}

\allowdisplaybreaks

\theoremstyle{plain}
\newtheorem{theorem}{Theorem}[section]

\newtheorem{lemma}[theorem]{Lemma}

\newtheorem{corollary}[theorem]{Corollary}
\theoremstyle{definition}

\newtheorem{observation}[theorem]{Observation}
\theoremstyle{remark}

\usepackage[textsize=tiny]{todonotes}
\usepackage{authblk}
\usepackage{lipsum}

\slimauthor{Smoothed Analysis of the $k$-Swap Neighborhood for Makespan Scheduling}

\begin{document}

\title{Smoothed Analysis of the $k$-Swap Neighborhood \\ for Makespan Scheduling\footnote{This research is supported by NWO grant OCENW.KLEIN.176.}}

\author[1]{Lars Rohwedder}
\author[2]{Ashkan Safari}
\author[2]{Tjark Vredeveld}
\affil[1]{Department of Mathematics and Computer Science, University of Southern Denmark, Denmark}
\affil[2]{Department of Quantitative Economics, School of Business and Economics, Maastricht University, The Netherlands}

\date{November, 2024}

\maketitle

\begin{abstract}

Local search is a widely used technique for tackling challenging optimization problems, offering simplicity and strong empirical performance across various problem domains. In this paper, we address the problem of scheduling a set of jobs on identical parallel machines with the objective of \emph{makespan minimization},
by considering a local search neighborhood, called \emph{k-swap}. A $k$-swap neighbor is obtained by interchanging the machine allocations of at most $k$ jobs scheduled on two machines.
While local search algorithms often perform well in practice, they can exhibit poor worst-case performance. In our previous study, we showed that for $k \geq 3$, there exists an instance where the number of iterations required to converge to a local optimum is exponential in the number of jobs. Motivated by this discrepancy between theoretical worst-case bound and practical performance, we apply \emph{smoothed analysis} to the $k$-swap local search. 
Smoothed analysis has emerged as a powerful framework for analyzing the behavior of algorithms, aiming to bridge the gap between poor worst-case and good empirical performance.
In this paper, we show that the smoothed number of iterations required to find a local optimum with respect to the $k$-swap neighborhood is bounded by $O(m^2 \cdot n^{2k+2} \cdot \log m \cdot \phi)$, where $n$ and $m$ are the numbers of jobs and machines, respectively, and $\phi \geq 1$ is the perturbation parameter.
The bound on the smoothed number of iterations demonstrates that the proposed lower bound reflects a pessimistic scenario which is rare in practice.

\end{abstract}
\section{Introduction}
Local search algorithms are among the most widely used heuristics to tackle computationally difficult optimization problems~\cite{aarts2003local, michiels2007theoretical}, due to their simplicity and strong empirical performance.
A local search procedure begins with an initial solution and iteratively moves from a feasible solution to a neighboring one until predefined stopping criteria are met. 
The selection of an appropriate neighborhood function significantly impacts the performance of the local search. These functions determine the range of solutions that the local search procedure can explore in a single iteration.
A classic example of local search is the \emph{2-opt} neighborhood for the traveling salesman problem~\cite{aarts2003local, croes1958method, flood1956traveling}, which often yields good solutions in practice despite the weak worst-case performance.
In some forms of local search, it is possible to move to a worse feasible solution within the neighborhood, which can help explore a broader range of solutions. However, in this paper, we focus on a basic form of local search, known as \emph{iterative improvement}, which always searches for a better solution within the neighborhood of the current one.
The iterative improvement stops when no further improvement can be achieved in the neighborhood, indicating that the current solution has reached a \emph{local optimum}.

In this paper, we consider one of the most fundamental scheduling problems, in which we are given a set of $n$ jobs with positive processing times $p_j (j = 1, 2, ..., n)$.
Each of these jobs needs to be processed on one of $m$ identical machines without preemption, so as to minimize the \emph{makespan}: we want the last job to be completed as early as possible.
This problem is denoted by $P{\parallel}C_{max}$~\cite{graham1979optimization}, and is strongly NP-hard~\cite{johnson1979computers}.
In consequence, many studies in the literature have suggested approximation algorithms to address this problem.
Among them, some of the first approximation algorithms~\cite{graham1966bounds, graham1969bounds} and polynomial time approximation schemes (PTAS)~\cite{hochbaum1987using} have been developed in the context of this problem.
As PTAS's are often considered impractical due to their tendency to yield excessively high running time bounds (albeit polynomial) even for moderate values of $\varepsilon$ (see~\cite{marx2008parameterized}), some studies have approached this problem by local search~\cite{brucker1997improving, finn1979linear, rohwedder2024kswap, schuurman2007performance}.
One of the simplest local search neighborhoods is the so-called \emph{jump} or \emph{move} neighborhood.
To obtain a jump neighbor, the machine allocation of one of the jobs changes. Finn and Horowitz~\cite{finn1979linear} proposed a simple improvement heuristic, called \textsc{0/1-Interchange} which is able to find a local optimum with respect to the jump neighborhood.
Brucker et al.~\cite{brucker1997improving} have provided an upper bound of $O(n^2)$ for this heuristic, and the tightness of this bound was shown by Hurkens and Vredeveld~\cite{hurkens2003local}. 
Another well-known folklore neighborhood is the \emph{swap} neighborhood, which is defined as the one consisting of all possible jumps and swaps~\cite{rohwedder2024kswap, schuurman2007performance}, where in a swap operator the machine allocations of two jobs, scheduled on two different machines, interchange (see Fig.~\ref{fig:swap-jump}).
In our previous work~\cite{rohwedder2024kswap}, we provided an upper bound of $O(n^4)$ for the number of swaps when we have only two machines in our schedule.
A natural generalization of the jump and swap neighborhoods is known as the \emph{$k$-swap}, which has been studied in this work~\cite{rohwedder2024kswap}. A $k$-swap neighbor is obtained by selecting at most $k$ jobs from two machines and interchanging the machine allocations of the selected jobs (see Fig.~\ref{fig:k-swap}).
Numerous alternative neighborhoods such as \emph{push}~\cite{schuurman2007performance}, \emph{multi-exchange}~\cite{ahuja2001multi, frangioni2004multi} and \emph{split}~\cite{brueggemann2011exponential} have been suggested in the literature to approach this problem.

\begin{figure}[t]
\centering
\includegraphics[width=0.7\textwidth]{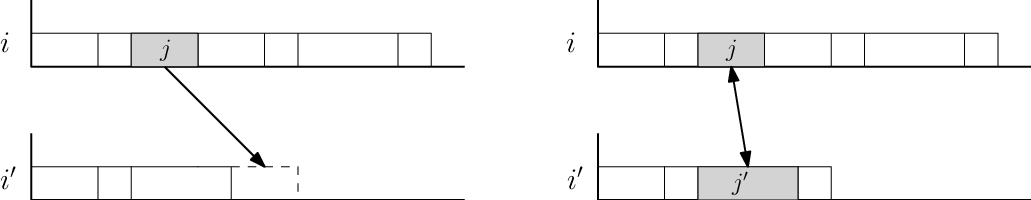}
\caption{Jump (left diagram) and swap (right diagram) operators.}
\label{fig:swap-jump}
\end{figure}

\afterpage{\begin{figure}[t!]
\centering
\includegraphics[width=0.40\textwidth]{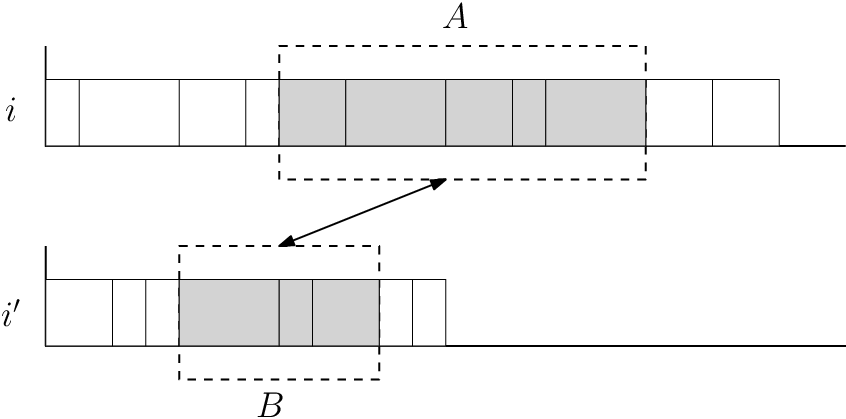}
\caption{$k$-swap operator.}
\label{fig:k-swap}
\end{figure}}

Furthermore, we provided a running time analysis of the $k$-swap neighborhood for the problem $P||C_{max}$ in our previous work~\cite{rohwedder2024kswap}. In this analysis, we showed that there exists an initial schedule as well as a sequence of 3-swaps such that the local optimum with respect to the $k$-swap neighborhood, for $k \geq 3$, is reached after $2^{\Omega(n)}$ iterations.
Moreover, our computational results on different classes of instances showed that, on average, the number of iterations required to converge to a local optimum is not excessively large.
This raises the question of whether the provided lower bound reflects a pessimistic scenario which might be rare in practice.


To address such discrepancies between worst-case and practical performance, the notion of \emph{smoothed analysis} was introduced by Spielman and Teng~\cite{spielman2004smoothed}.
Smoothed analysis is a hybrid of worst-case and average-case analysis, designed to provide a more realistic measure of an algorithm's performance.
In the \emph{classical} model of the smoothed analysis (also known as the \emph{two-step} model), an adversary selects an arbitrary instance and then, the numbers or coefficients of this instance are (slightly) randomly perturbed by adding Gaussian-distributed random variables of mean 0 and standard deviation $\sigma > 0$ to them.
The smoothed performance is then defined as the worst expected performance among all adversarial choices; the expected value is taken over the random perturbation of the adversarial instance.
In this setting, the smaller $\sigma$, the more powerful the adversary, i.e., the analysis is closer to the worst-case analysis.
In a more general model of the smoothed analysis, known as the \emph{one-step} model, introduced by Beier and Vöcking~\cite{beier2003random}, an adversary specifies probability density functions upper-bounded by $\phi \geq 1$ for all the numbers that need to be perturbed.
Then, all the numbers are drawn independently according to their respective density function.
In this setting, the smaller $\phi$, the less powerful the adversary, i.e., the analysis is closer to the average-case analysis.
The one-step model of the smoothed analysis contains the two-step model as a special case.
Smoothed analysis was first applied to the simplex method~\cite{spielman2004smoothed}, and was successfully able to explain why this method is performing well in practice. Subsequently, it was applied to different problems and algorithms, including linear programming~\cite{blum2002smoothed, sankar2006smoothed, spielman2003smoothed}, searching and sorting~\cite{banderier2003smoothed, damerow2012smoothed, fouz2012smoothed, manthey2007smoothed}, online and approximation algorithms~\cite{becchetti2006average, blaser2013smoothed, schafer2005topology}, game theory~\cite{boros2011stochastic, chen2009settling}, and local search~\cite{arthur2006worst, englert2014worst, manthey2015smoothed, manthey2022improved}.

Motivated by the exponential lower bound established in our previous work~\cite{rohwedder2024kswap}, we perform smoothed analysis on the number of $k$-swaps for the problem $P||C_{max}$, and show that it is upper bounded by $O(m^2 \cdot n^{2k+2} \cdot \log m \cdot \phi)$, where $\phi \geq 1$ is the perturbation parameter.
While the lower bound provided in our previous study~\cite{rohwedder2024kswap} is exponential in the number of jobs, our result in this paper demonstrates a polynomial upper bound for the smoothed number of $k$-swaps in obtaining a local optimum when $k$ is constant.

\section{Preliminaries and Notation}
Recall that in the problem under consideration, there are $n$ jobs that need to be processed without interruption, on $m$ identical machines.
A job can only be processed by one machine and each machine can process at most one job at a time. The time it takes to process job $j$ is denoted by $p_j$.
The goal is to schedule the jobs such that the makespan, $C_{\max} = \max_j C_j$, is minimized, where $C_j$ denotes the time at which job $j$ is completed.
As the order in which the jobs, assigned to the same machine, are processed, does not influence the makespan, we represent a schedule by $(M_1, M_2, ... , M_m)$, where $M_i$ denotes the set of jobs scheduled on machine $i$, for $i \in \{1, 2, ... , m\}$.
Let $L_i = \sum_{j \in M_i} p_j$ denote the load of machine $i$, and $p(S) = \sum_{j \in S} p_j$ denote the summation of the processing times of a set of jobs $S$.
It is easy to see as there is no advantage in keeping a machine idle while it still needs to process at least a job, the makespan of a schedule is equal to the maximum load $L_{\max}$, i.e., $C_{\max} = L_{\max}$; we will use $C_{\max}$ and $L_{\max}$ interchangeably.
A \emph{critical machine} $i$ is a machine with $L_i = L_{\max}$, and a \emph{non-critical machine} $i'$ is a machine with $L_{i'} < L_{\max}$.
Additionally, we define an \emph{$\ell$min-load machine} as the machine with the $\ell$th minimum load among all the machines. We use the notation $L_{\ell\min}$ to represent the load of this machine.
Moreover, we denote by $L_{\max}(t), L_{i}(t)$, $L_{\ell\min}(t)$ and $M_i(t)$ the values of the corresponding variables at the start of iteration $t$.

The $k$-swap neighborhood is defined as the one consisting of swapping all combinations of at most $k$ jobs between any pair of machines. To be more precise, to obtain a $k$-swap neighbor, a set of jobs $A$ and a set of jobs $B$ are selected from machines $i$ and $i'$, respectively, such that $|A| \geq 1$, $|B| \geq 0$ and $|A| + |B| \leq k$, and then the machine allocations of these jobs are exchanged (see Fig.~\ref{fig:k-swap}). 
In order to improve the schedule, $i$ needs to be a critical machine, and $0 < p(A) - p(B) < L_i - L_{i'}$.
We say we are in a \emph{k-swap optimal solution} if, by applying the $k$-swap operator, we cannot decrease the makespan or the number of critical machines without increasing the makespan. 


\section{Smoothed Analysis}
In this section, we provide a smoothed analysis of the number of iterations required to converge to a $k$-swap optimal solution. 
To this end, we employ the one-step smoothing model~\cite{beier2003random}. 
In this model, an adversary specifies probability density functions upper-bounded by $\phi \geq 1$ for all the numbers that need to be perturbed. Then, all the numbers are drawn independently according to their respective density function.
Note that the perturbation parameter $\phi$ controls the adversary's power; the larger the value of $\phi$, the more powerful the adversary is.
Finally, we measure the maximum expected performance, where the maximum is taken over all the choices of the adversary, and the expectation is taken over drawing the numbers according to their density functions.
In our analysis, we perturb the processing times of the jobs. An adversary specifies density functions $f_1, \dotsc, f_n$ with $f_i : [0, 1] \rightarrow [0, \phi]$. An instance of the problem is obtained by drawing $p_1, \dotsc, p_n$ independently according to the densities $f_1, \dotsc, f_n$, respectively.

\begin{figure}[t!]
\centering
\includegraphics[width=0.95\textwidth]{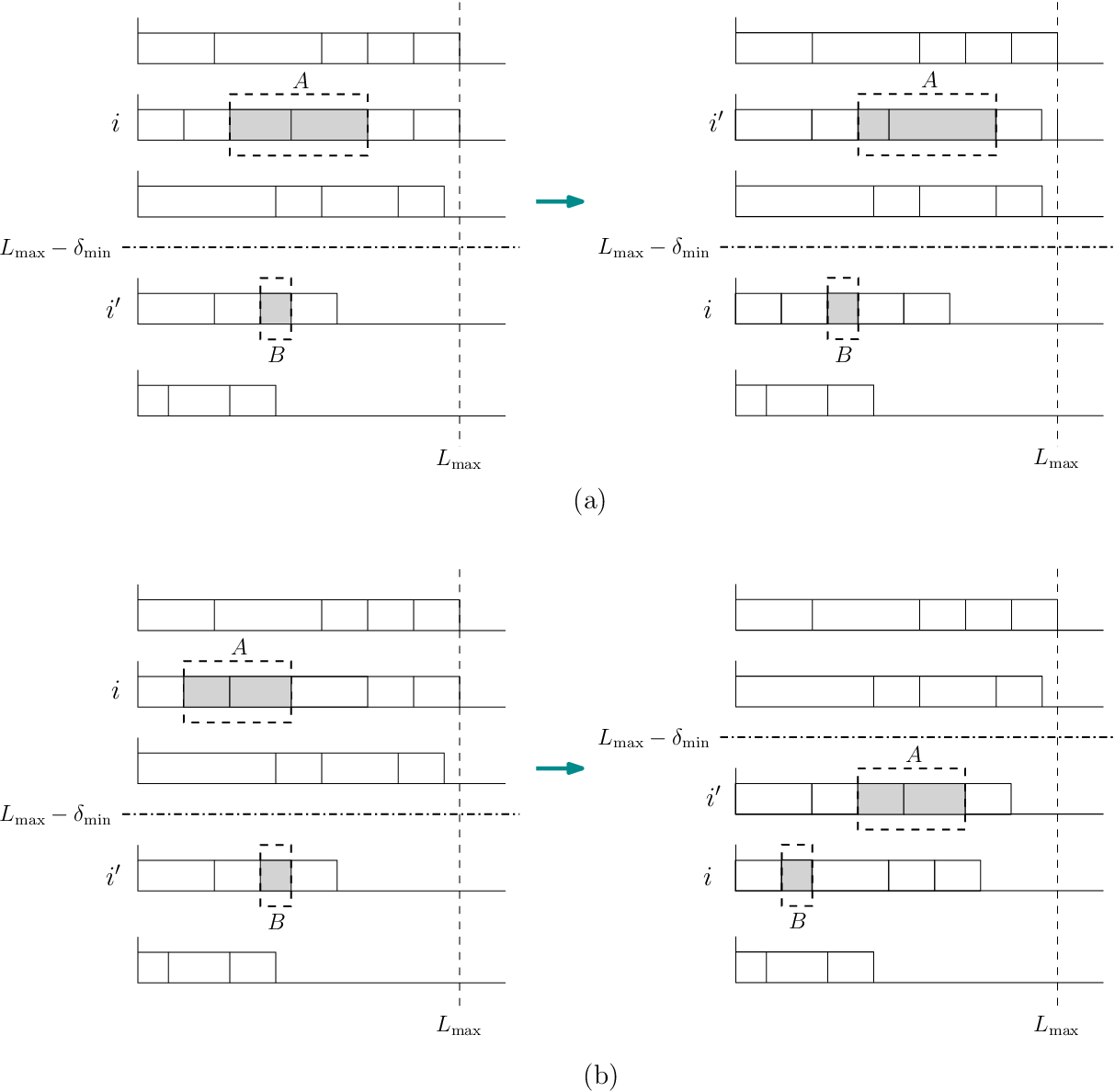}
\caption{(a) Type-1 of improving $k$-swap. After swapping $A$ with $B$, $i'$ moves to $\gamma_l$. (b) Type-2 of improving $k$-swap. After swapping $A$ with $B$, $i'$ stays in $\gamma_s$. Note that in both of the cases, machine $i$ does not necessarily move to $\gamma_s$ after the swap.}
\label{fig:k-swapiPrime}
\end{figure}

Let $\delta_{min} = \min | p(A) -p(B) |$, where the minimum is taken over all the sets of jobs $A$ and $B$ with $A \cap B = \varnothing$ and $1 \leq |A| + |B| \leq k$.
To perform the analysis, we split the set of machines in our schedule at the start of iteration $t$ into two sets $\gamma_l(t) = \{ i: L_i(t) > L_{max}(t) - \delta_{min} \}$ and $\gamma_s(t) = \{ i: L_i(t) \leq L_{max}(t) - \delta_{min} \}$.
We refer to these two sets as $\gamma_l$ and $\gamma_s$ if $t$ is clear from the context.
In order to improve the schedule in iteration $t$, we need to swap a set of jobs $A$ from a critical machine $i$ with a set of jobs $B$ from a non-critical machine $i'$ such that $|A| + |B| \leq k$, and $0 < p(A) - p(B) < L_{i}(t) - L_{i'}(t)$.
Since all the critical machines are in the set $\gamma_l(t)$ and loads of all the machines in $\gamma_l(t)$ are greater than $L_{max}(t) - \delta_{min}$, we always have $i \in \gamma_l(t)$ and $i' \in \gamma_s(t)$ in an \emph{improving} iteration $t$ (an iteration in which we find an improving $k$-swap neighbor).

In our analysis, we consider two types of improving iterations. 
In an improving $k$-swap of \emph{type-1}, machine $i'$ moves to $\gamma_l$ after swapping $A$ with $B$ in iteration $t$ (see Figure~\ref{fig:k-swapiPrime} (a)); formally, we get $i' \in \gamma_l(t+1)$. 
In an improving $k$-swap of \emph{type-2}, machine $i' $ stays in $\gamma_s$ after swapping $A$ with $B$ in iteration $t$ (see Figure~\ref{fig:k-swapiPrime} (b)); formally, we get $i' \in \gamma_s(t+1)$.
We analyze this type based on a fixed value of $\delta_{min}$, by employing a potential function and counting the number of times that this function decreases. 
We will provide an upper bound of $O(m \cdot n^k)$ for the number of consecutive iterations of type-1, and an upper bound of $O(\dfrac{m \cdot n}{\delta_{\min}})$ for the total number of iterations of type-2.
Therefore, the total number of iterations in obtaining a $k$-swap optimal solution is upper-bounded by $O(\dfrac{m^2 \cdot n^{k+1}}{\delta_{min}})$.
In the worst case, $\delta_{min}$ can be (close to) 0, resulting in a very large bound; however, in the smoothed setting, the probability that this happens is very small, and we can show a bound of $O(m^2 \cdot n^{2k+2} \cdot \log m \cdot \phi)$.

First, we show that in a sequence of consecutive $k$-swaps of type-1, the sorted vector of machine loads in the set $\gamma_s$ is pointwise non-decreasing.


\begin{lemma}
    \label{lem:llminMonotone}
    For all the $\ell$min-load machines in the set $\gamma_s$, we have $L_{\ell\min} (t+1) \geq L_{\ell\min} (t)$ when a $k$-swap operation of type-1 occurs in iteration $t$.
\end{lemma}
\begin{proof}
    Let $i \in \gamma_l(t)$ be a critical machine and $i' \in \gamma_s(t)$, and let $A \subseteq M_i(t)$ and $B \subseteq M_{i'}(t)$ with $0 < p(A) - p(B) < L_{i}(t) - L_{i'}(t)$ be two sets of jobs such that swapping them constitute an improving $k$-swap of type-1; i.e., we get $i' \in \gamma_l(t+1)$.
    After doing the swap, machine $i$ either stays in $\gamma_l$ or moves to $\gamma_s$. As we have $i' \in \gamma_l(t+1)$, and the values of $L_{\max} - \delta_{\min}$ are non-increasing, $|\gamma_s|$ does not increase.
    If after the swap, we have $i \in \gamma_l(t+1)$, the load of none of the machines in $\gamma_s$ decreases.
    If after the swap, we have $i \in \gamma_s(t+1)$, since $p(A) - p(B) < L_{i}(t) - L_{i'}(t)$, we get $L_{i}(t+1) > L_{i'}(t)$, and therefore, the lemma holds.
\end{proof}

Next, we show that in a consecutive sequence of $k$-swaps of type-1 when a set $A$ of jobs on a critical machine is swapped with a set $B$ on an $\ell$min-load machine, then these two sets cannot be swapped again on a critical and an $\ell$min-load machine in the same sequence of consecutive $k$-swaps of type-1.

\begin{lemma}
    \label{lem:fixTwojobs2}
    Consider an improving $k$-swap of type-1 by swapping a set $A$ on a critical machine $i$ with a set $B$ on an $\ell$min-load machine $i'$ in iteration $t$. Then, we are not able to swap $A$ with $B$ on a critical machine and an $\ell$min-load machine as long as we have not performed the $k$-swap operation of type-2.
\end{lemma}
\begin{proof}
    By swapping the jobs in $A$ with the jobs in $B$, and according to Lemma~\ref{lem:llminMonotone}, we get:
    \begin{align*}
        L_{i'}(t+1) - L_{\ell\min}(t+1)  &\phantom{}=  L_{i'}(t) + p(A) - p(B) - L_{\ell\min}(t+1) \\
        &\leq L_{i'}(t) + p(A) - p(B) - L_{\ell\min}(t) \\
        &= L_{i'}(t) + p(A) - p(B) - L_{i'}(t) = p(A) - p(B).
    \end{align*}
    In iteration $t+1$, the jobs of set $A$ are in machine $i'$ and machine $i'$ is in $\gamma_l(t+1)$. 
    The jobs of set $A$ can be swapped with other jobs when machine $i'$ becomes a critical machine. If $i'$ becomes a critical machine in iteration $t'$ ($t' \geq t+1$), we get $L_{\max}(t') = L_{i'}(t') = L_{i'}(t+1)$, because $i' \in \gamma_{\ell}(t+1)$, and therefore there cannot be an improving $k$-swap with jobs on machine $i'$, as long as $i'$ is not a critical machine. 
    As we have, $L_{\ell\min}(t') \geq L_{\ell\min}(t+1)$, we get
    \begin{align*}
       L_{\max}(t') - L_{\ell\min}(t') \leq L_{i'}(t+1) - L_{\ell\min}(t+1) \leq p(A) - p(B).
    \end{align*}
    According to Lemma~\ref{lem:llminMonotone}, in a consecutive sequence of improving $k$-swaps of type-1, the values of $L_{\max} - L_{\ell\min}$ are non-increasing, and we cannot swap $A$ with $B$ on a critical machine and an $\ell$min-load machine in the remainder of the sequence of consecutive $k$-swaps of type-1.
\end{proof}

According to Lemma~\ref{lem:fixTwojobs2}, by performing a $k$-swap operation of type-1 over the sets $A \subseteq M_i(t)$ and $B \subseteq M_{i'}(t)$, such that $i'$ is an $\ell$min-load machine at the start of iteration $t$, we are not able to swap the jobs in $A$ with the jobs in $B$ on a critical machine and an $\ell$min-load machine as long as we have not performed the $k$-swap operation of the type-2. 
As we have at most $O(m \cdot n^k)$ different combinations of $A$, $B$, and $\ell$ (we have $O(n^k \cdot 2^k)$ different combinations of the sets $A$ and $B$), the number of consecutive $k$-swaps of type-1 is upper-bounded by $O(m \cdot n^k)$.

\begin{corollary}
    \label{cor:type1and2}
    In at most $O(m \cdot n^k)$ consecutive iterations, the $k$-swaps of the type-1 occur.
\end{corollary}

To analyze the number of $k$-swaps of type-2, we make use of the potential function defined by $\Phi = \sum\limits_{i,i'} |L_i - L_{i'}|$. 
Assume w.l.o.g. that $L_1 \geq L_2 \geq \dotsc \geq L_m$. Then, $\Phi = 2 \sum_{1 \leq i < i' \leq m} (L_i - L_{i'}) = 2\sum_{1 \leq i < i' \leq m} (m-2i+1)L_i$. 
As we know that the total sum of machine loads is upper-bounded by $n$, we get $\Phi \leq 2m \cdot n$.

\begin{observation}
    \label{obs:phiUpperBound}
    We have $\Phi(1) \leq 2 m \cdot n$.
\end{observation}

Next, we provide a lower bound for the amount that $\Phi$ decreases in each iteration of this case. To do so, first, we provide the following lemma.

\begin{lemma}
    \label{lem:graterThanZeroInequality}
     Let $a, b, c\in \mathbb{R}$ such that $0 < c < a - b$. We have $|a - c| + |b + c| \leq |a| + |b|$.
\end{lemma}
\begin{proof}
Let $\lambda = \frac{c}{a-b}$. Since $0 < c < a - b$, we have $0 < \lambda < 1$. As the absolute value is a convex function, we have
    \[
        |a - c| + |b + c|
        = |(1 - \lambda)a + \lambda b| + |(1 - \lambda)b + \lambda a|
        \leq (1 - \lambda)|a| + \lambda |b|  + (1 - \lambda) |b| + \lambda |a|
        = |a| + |b|. \qedhere
    \]
\end{proof}

Furthermore, since a $k$-swap of type-2 between a critical machine $i$ and a machine $i'$ in iteration $t$ results in a load for $i'$ that is at least $\delta_{min}$ less than the makespan, i.e., $L_{i'}(t+1) \leq L_{\max}(t+1) - \delta_{\min}$, we have the following observation.

\begin{observation}
    \label{obs:deltaMinDecrease}
    Consider an improving $k$-swap of type-2 by swapping a set $A$ on a critical machine $i$ with a set $B$ on a machine $i'$. Then, we have $(L_i(t) - L_{i'}(t)) - (p(A) -  p(B)) \geq \delta_{min}$.
\end{observation}

Subsequently, we show that by performing a $k$-swap operation of type-2, the value of $\Phi$ decreases by at least $4\delta_{\min}$.

\begin{lemma}
    \label{lem:PhiLowerBound}
    Consider an improving $k$-swap of type-2 in iteration $t$.
    We have $\Phi(t) - \Phi(t+1)  \geq 4\delta_{\min}$.
\end{lemma}
\begin{proof}
    Assume the $k$-swap of type-2 is obtained by swapping a set $A$ on a critical machine $i$ with a set $B$ on a machine $i'$ in iteration $t$.
    Let $\Delta_{\Phi} = \Phi(t) - \Phi(t+1)$. We have
    \begin{align*}
        \Delta_{\Phi} &= 2|L_i(t) - L_{i'}(t)| - 2|L_i(t) - L_{i'}(t) - 2(p(A) -  p(B))| \\
        &\quad+ 2\sum_{i'' \neq i,i'}|L_i(t) - L_{i''}(t)| - 2\sum_{i'' \neq i,i'}|L_i(t) - p(A) + p(B) - L_{i''}(t)| \\
        &\quad+ 2\sum_{i'' \neq i,i'}|L_{i'}(t) - L_{i''}(t)| - 2\sum_{i'' \neq i,i'}|L_{i'}(t) + p(A) -  p(B) - L_{i''}(t)|.
    \end{align*}
    As swapping $A$ with $B$ in iteration $t$ yields an improving $k$-swap, we know that $0 < p(A) - p(B) < L_{i}(t) - L_{i'}(t)$.
    Hence, by Lemma~\ref{lem:graterThanZeroInequality}, we have $|L_i(t) - L_{i''}(t)| + |L_{i'}(t) - L_{i''}(t)| \geq |L_i(t) - L_{i''}(t) - (p(A) - p(B))| + |L_{i'}(t) - L_{i''}(t) + p(A) -  p(B)|$ .
    Therefore, 
    \begin{align*}
        \Delta_{\Phi} \geq 2|L_i(t) - L_{i'}(t)| - 2|L_i(t) - L_{i'}(t) - 2(p(A) - p(B))|.
    \end{align*}
    If $p(A) - p(B) \leq \frac{L_i(t) - L_{i'}(t)}{2}$, then $|L_i(t) - L_{i'}(t) - 2(p(A) - p(B))| = L_i(t) - L_{i'}(t) - 2(p(A) - p(B))$ and we have
    \begin{align*}
        \Delta_{\Phi} \geq 2(L_i(t) - L_{i'}(t)) - 2(L_i(t) - L_{i'}(t) - 2(p(A) - p(B))) = 4(p(A) - p(B)) \geq 4\delta_{\min}.
    \end{align*}
    On the other hand, if $p(A) - p(B) > \frac{L_i(t) - L_{i'}(t)}{2}$, then $|L_i(t) - L_{i'}(t) - 2(p(A) - p(B))| = 2(p(A) - p(B)) - (L_i(t) - L_{i'}(t))$. By Observation~\ref{obs:deltaMinDecrease}, we have
    \[
        \Delta_{\Phi} \geq 2(L_i(t) - L_{i'}(t)) + 2(L_i(t) - L_{i'}(t) - 2(p(A) - p(B))) = 4(L_i(t) - L_{i'}(t)) - 4(p(A) - p(B)) \geq 4\delta_{\min}. \qedhere
    \]
\end{proof}

Using Observation~\ref{obs:phiUpperBound} and Lemma~\ref{lem:PhiLowerBound}, we can bound the number of improving $k$-swaps of type-2 as in Corollary~\ref{cor:type3}.

\begin{corollary}
    \label{cor:type3}
    In at most $O(\dfrac{m \cdot n}{\delta_{\min}})$ iterations, the $k$-swaps of type-2 occur.
\end{corollary}

As the consecutive number of improving $k$-swaps of type-1 is upper-bounded by $O(m \cdot n^k)$ (Corollary~\ref{cor:type1and2}), we can bound the total number of iterations as in Corollary~\ref{cor:totalkSwap}.

\begin{corollary}
    \label{cor:totalkSwap}
    The total number of iterations in obtaining a $k$-swap optimal solution is upper bounded by $O(\dfrac{m^2 \cdot n^{k+1}}{\delta_{\min}})$.
\end{corollary}

To bound the expected number of iterations in the smoothed analysis, we bound the probability that $\delta_{\min}$ is small.

\begin{lemma}
\label{lem:deltaMinProb2}
    Let $\alpha \in (0, k]$. Then, $\probP(\delta_{min} \leq \alpha) \leq 2^{k+1} \cdot n^k \cdot \alpha \cdot \phi$.
\end{lemma}
\begin{proof}
    Let $A$ and $B$ be two disjoint subsets of jobs such that $1 \leq |A| + |B| \leq k$. By the principle of deferred decisions, we have that $\probP(|p(A) - p(B)| \leq \alpha) =  2 \cdot \alpha \cdot \phi$.
    By taking a union bound over the $(2n)^k$ possibilities for the sets $A$ and $B$, we get the desired result.
\end{proof}




\begin{theorem}
    \label{thm:expected2Swap}
    Let $f_1, \dotsc, f_n :  [0, 1] \rightarrow [0, \phi]$ be arbitrary density functions bounded by $\phi \geq 1$, and let $p_1, \dotsc, p_n$ be the processing times of the jobs drawn independently according to $f_1, \dotsc, f_n$, respectively. Let $T$ denote the random variable of the number of iterations to find a $k$-swap optimal solution. Then
    \begin{center}
        $\exP(T) = O(m^2 \cdot n^{2k+2} \cdot \log m \cdot \phi).$
    \end{center}
\end{theorem}
\begin{proof}
    Observe that, by Corollary~\ref{cor:totalkSwap}, when $T \geq t$, for a given $t$, we need to have that $\delta_{min} \leq \dfrac{m^2 \cdot n^{k+1}}{t}$. Hence, we can bound the expected value of $T$ by the number of different schedules. Therefore,
    \begin{align*}
        \exP(T) = \sum_{t=1}^{m^n} \probP(T \geq t) = \sum_{t=1}^{m^n} \probP(\delta_{min} \leq \dfrac{m^2 \cdot n^{k+1}}{t})
        &\leq \sum_{t=1}^{m^n} \dfrac{2^{k+1} \cdot m^2 \cdot n^{2k+1} \cdot \phi}{t} \\
        &= 2^{k+1} \cdot \phi \cdot m^2 \cdot n^{2k+1} \cdot \sum_{t=1}^{m^n} \dfrac{1}{t} \\
        &= O(m^2 \cdot n^{2k+2} \cdot \log m \cdot \phi),
    \end{align*}
    Where the inequality is according to Lemma~\ref{lem:deltaMinProb2}. 
\end{proof}

\bibliographystyle{plain}
\bibliography{ref}


\end{document}